\documentclass[letterpaper, 11 pt,onecolumn,draftclsnofoot]{IEEEtran}
\usepackage[T1]{fontenc}
\usepackage[latin1]{inputenc}
\usepackage{float}
\usepackage{amsmath}
\usepackage{cite}
\usepackage{amssymb}
\usepackage{enumerate}
\usepackage{amsthm}
\usepackage{graphics,epsfig}
\usepackage{subfigure,epsf,pstricks}

\floatstyle{ruled}
\newfloat{algorithm}{tbp}{loa}
\floatname{algorithm}{Algorithm}

\usepackage{algorithmic}
\algsetup{linenodelimiter= }

\newtheorem{theorem}{Theorem}
\newtheorem{lemma}{Lemma}
\newtheorem{cor}{Corollary}

\newtheorem{defi}{Definition}
\newtheorem{prop}{Proposition}

\numberwithin{equation}{section}

\newcommand{\R}{\mathbb{R}}

\newcommand{\rank}{\operatorname{rank}}

\newcommand{\minimize}{\mbox{minimize}}

\newcommand{\st}{\mbox{subject to}}

\newcommand{\cA}{\mathcal{A}}

\begin{document}

\title{A Unique ``Nonnegative'' Solution to an Underdetermined System: from Vectors to Matrices}
\author{Meng Wang \ \ \ \ \ Weiyu Xu \ \ \ \ \ Ao Tang
\thanks{The authors are with School of Electrical and Computer
Engineering, Cornell University, Ithaca, NY, 14853. \{mw467, wx42,
at422\}@cornell.edu. }}

\maketitle 

\begin{abstract}

This paper investigates the uniqueness of a nonnegative vector
solution and the uniqueness of a positive semidefinite matrix
solution to underdetermined linear systems. A vector solution is the unique solution to an underdetermined linear system only if the measurement matrix has a row-span intersecting the positive
orthant. 
Focusing on two types of binary measurement matrices, Bernoulli 0-1
matrices and adjacency matrices of general expander graphs, we show
that, in both cases, the support size of a unique nonnegative solution can grow
linearly, namely $O(n)$, with the problem dimension $n$.
We also provide closed-form characterizations of the ratio of this
support size to the signal dimension.
For the matrix case, we show that under a necessary and sufficient
condition for the linear compressed observations operator, there
will be a unique positive semidefinite matrix solution to the
compressed linear observations. 
We further show that a randomly generated Gaussian linear compressed
observations operator will satisfy this condition with overwhelmingly high probability.

\end{abstract}

\section{Introduction}

This paper is devoted to recover a ``nonnegative'' decision variable
from an underdetermined system of linear equations. When the
decision variable is a vector, ``nonnegativity'' means each entry is
nonnegative. When the decision variable is a matrix,
``nonnegativity'' indicates that the matrix is positive
semidefinite. The problem is ill-conditioned in general, however, we
can correctly recover the vector or the matrix if the vector is
sparse, or the matrix is low rank.

 Finding the sparest vector among vectors satisfying a set of linear
 equations is NP-hard. One frequently used heuristic is \textit{$L_1$-minimization}, which returns the vector with the least $L_1$ norm.
Recently, there has been an explosion of research on this topic, see e.g.,
\cite{Can06,Bar07,Don06,CaT06,CaT05}. 
\cite{CaT05} gives a sufficient condition known as Restricted
Isometry Property (RIP) on the measurement matrix that guarantees
the recovery of the sparest vector via $L_1$ minimization. In many
interesting cases, the vector is known to be nonnegative.
\cite{DoT05} gives a necessary and sufficient condition known as the
outwardly neighborliness property of the measurement matrix for
$L_1$ minimization to successfully recover a sparse non-negative
vector.
Moreover, recent studies \cite{BEZ08,DoT08,KDXH09} suggested that a
sparse solution could be the unique nonnegative solution there. This
certainly leads to potentially better alternatives to $L_1$
minimization as in this case any optimization problem over this
constraint set can recover the solution. 

Motivated by networking inference problems such as network
tomography
, we are particularly interested in systems where the measurement
matrix is a 0-1 matrix. There have not been many existing results on
this type of systems except a few very recent papers
\cite{XuH07,BGIKS08,BeI08,KDXH09}. We focus on two types of binary
matrices, Bernoulli 0-1 matrices and adjacency matrices of
expanders, and provides conditions under which a sparse vector
is the unique nonnegative solution to the underdetermined system.
For random Bernoulli measurement matrices, we prove that, as long as the number of equations divided by the number of variables remains constant as the problem dimension grows, with overwhelming probability over the choices of matrices, a sparse
nonnegative vector is a unique nonnegative solution provided that
its support size is at most proportional to its dimension for some
positive ratio. For general expander matrices, we further provide a
closed-form constant ratio of support size to dimension under which
a nonnegative vector is the unique solution.

The phenomenon that an underdetermined system admits a unique
``nonnegative'' solution is not restricted for the vector case.
Finding the minimum rank matrix among all matrices satisfying
given linear equations is a \emph{rank minimization}
problem. 
Among the rank minimization problems, one particularly important class is the rank minimization problem for positive semidefinite matrices under compressed observations. For example, minimizing the rank of a covariance matrix, which is a positive semidefinite matrix, arises in statistics, econometrics, signal processing and many other fields where second-order statistics for random processes are used \cite{FazelThesis}.  
A positive semidefinite matrix is special in that its eigenvalues
(also its singular values) are nonnegative.  In fact, the nuclear
norm minimization heuristic for general matrices was preceded by the
trace norm heuristic for positive symmetric matrices in rank
minimization problems.  While the general analytic frameworks and
computational techniques, for example, \cite{Recht07, RechtCDC08},
are applicable to the rank minimization problems for positive
semidefinite matrices, the special properties of positive
semidefinite matrices may open the way to new structures and new
analysis, which more efficient computational techniques may exploit
to provide faster matrix recovery.

Parallel to the influence of the nonnegative constraint on a vector
variable, the positive semidefinite constraint on a matrix variable
may dramatically reduce the size of the feasible set in rank
minimization problems. In particular, we show that under a necessary
and sufficient condition for the linear compressed observations
operator, there will be a unique positive semidefinite matrix
solution to compressed linear observations. We further show 
that a randomly generated Gaussian linear compressed observations
operator will satisfy this necessary and sufficient condition with
overwhelmingly high probability. This result is akin to the one in
the vector case for the unique nonnegative solution,
but the transition from a nonnegative vector to a positive
semidefinite matrix requires very different analytical approaches.

This paper is organized as follows. Section \ref{sec:vector}
discusses the 
phenomena that a sparse vector can be the unique nonnegative vector
satisfying an underdetermined linear system. 
Focusing on 0-1 matrices, we prove that a sparse vector is a unique
nonnegative solution as long as its support size is at most
proportional to the dimension for some positive ratio.
We further give a closed-form ratio of the support size and the
dimension if the matrix is an adjacent matrix of an expander graph.
Section \ref{sec:matrix} shows a low-rank matrix can be the unique
positive semidefinite matrix satisfying compressed linear
measurements. We provide a necessary and sufficient condition for
this phenomenon to happen and prove the existence of compressed
measurements satisfying the proposed condition. 
Numerical examples
are discussed in Section \ref{sec:simulation} and Section
\ref{sec:conclusion} concludes the paper.


\section{Unique Nonnegative Vector to an Underdetermined
System}\label{sec:vector}
How to recover a vector $x \in \R ^n$ from the measurement $y=Ax \in
\R^m$, where $A^{m \times n}$($m<n$) is the measurement matrix? In
many applications, $x$ is nonnegative, which is our main focus here.
In general, the task seems impossible as we have fewer measurements
than variables. However, if $x$ is sparse, it can be recovered by
solving the following problem,
\begin{eqnarray}\label{eqn:l0}
 \min  \|x\|_0 &&
 \textrm{s.t. }  Ax=y, x \geq 0 ,
\end{eqnarray}
where the $L_0$ norm $\|\cdot\|_0$ measures the number of nonzero
entries of a given vector. Since (\ref{eqn:l0}) in general is
NP-hard, people solve an alternative convex problem by replacing
$L_0$ norm with $L_1$ norm where $\|x\|_1= \sum_i |x_i|$.
The $L_1$ minimization problem can be formulated as follows:
\begin{eqnarray}\label{eqn:l1}
 \min  \mathbf{1}^Tx &&
 \textrm{s.t. } Ax=y, x \geq 0.
\end{eqnarray}

In fact, for a certain class of matrices, if $x$ is sufficiently
sparse, not only can we recover $x$ from (\ref{eqn:l1})
, but also $x$ is the \emph{only} solution to $\{x ~|~ Ax= y, x \geq 0\}$.
In other words, $\{x ~|~ Ax= y, x \geq 0\}$ is a singleton, and $x$
can possibly be recovered by techniques other
than $L_1$ minimization. 
%


\cite{BEZ08} analyzed the singleton property of matrices with a
row-span intersecting the positive orthant. Here we first show only
these matrices can possibly have the singleton property.

\begin{defi}[\cite{BEZ08}]  $A$ has a row-span intersecting the positive orthant, denoted by $A \in
\mathbf{M}^+$, if there exists a vector $\beta>0$ in the row space
of $A$, i.e. $\exists h $ such that
\begin{equation}\nonumber
h^T A= \beta^T >0.
\end{equation}
\end{defi}

There is a simple observation regarding matrices in $\mathbf{M}^+$.

\begin{lemma} \label{lem:M} Let $a_i \in \mathbb{R}^m$ $(i=1,2,...,n)$ be the
$i^{\textrm{th}}$ column of matrix $A$, then $A \in \mathbf{M}^+$ if
and only if $0 \notin P$, where

\begin{equation}\label{eqn:polytope} \nonumber
P \triangleq\textbf{Conv}(a_1,a_2,...,a_n)=\{\sum_{i}
\lambda_i a_i | \mathbf{1}^T \lambda=1, \lambda \geq 0 \}
\end{equation}
\end{lemma}
\begin{proof} If $A \in \mathbf{M}^+$, then there exists $h$
such that $h^T A= \beta^T >0$. Suppose we also have $ 0 \in P$, then
there exists $\lambda \geq 0$ such that $A\lambda=0$ and
$\mathbf{1}^T \lambda=1$. Then $(h^T A) \lambda= \beta^T \lambda >0$
as $\beta>0$, $\lambda \geq 0$ and $\lambda \neq 0$. But $(h^T A)
\lambda=h^T (A \lambda)=0$ as $A\lambda=0$. Contradiction! Therefore
$ 0 \notin P$.

Conversely, if $ 0 \notin P$, there exists a separating hyperplane
$\{x~|~ h^Tx+b=0, h \neq 0\}$ that strictly separates $0$ and $P$.
We assume without loss of generality that $h^T0+b <0$ and $h^Tx
+b>0$ for any point $x$ in $P$. Then $h^T a_i > -b> 0, \forall i$.
Thus we conclude $h^T A
>0$.

\end{proof}

The next theorem states a necessary condition on 
matrix $A$ for $\{x ~|~ Ax= Ax_0, x
\geq 0\}$ to be a singleton. 

\begin{theorem}\label{thm:general} If $\{x ~|~ Ax= Ax_0, x \geq
0\}$ is a singleton for some $x_0 \geq 0$, then $A \in
\mathbf{M}^+$.
\end{theorem}
\begin{proof} Suppose $A \notin \mathbf{M}^+$, from Lemma \ref{lem:M} we know $0 \in
\textbf{Conv}(a_1,a_2,...,a_n)$. Then there exists a vector $w \geq
0$ such that $Aw=0$ and $\mathbf{1}^Tw=1$. Clearly $w \in
\textbf{Null}(A)$ and $w \neq 0$. Then for any $\gamma>0$ we have
$A(x_0+\gamma w)=Ax_0+\gamma Aw= Ax_0$, and $x_0+\gamma w \geq 0$
provided $x_0 \geq 0$. Hence $x_0+\gamma w \in \{x ~|~ Ax= Ax_0, x
\geq 0\}$.

\end{proof}

Theorem \ref{thm:general} shows that $A$ belongs to $\mathbf{M}^+$
is a necessary condition for an underdetermined system to admit a unique nonnegative vector. 
If $A^{m \times n}$ is a random matrix such that every entry is
independently sampled from Gaussian distribution with zero mean,
then the probability that $0$ lies in the convex hull of the column
vectors of $A$, or equivalently $\{x ~|~ Ax= Ax_0, x \geq 0\}$ is
not a singleton for any $x_0 \geq 0$, is $1-2^{-n+1}\sum
\limits_{k=0}^{m-1} {{n-1}\choose{k}}$(\cite{Wendel62}), which goes
to 1 asymptotically as $n$ increases if $\lim \limits_{n \rightarrow
+ \infty}
\frac{m}{n}<\frac{1}{2}$. 
Thus, if $\lim \limits_{n \rightarrow + \infty}
\frac{m}{n}<\frac{1}{2}$, then for a random Gaussian matrix $A$,
$\{x ~|~ Ax= Ax_0, x \geq 0\}$ would not be a singleton with
overwhelming probability no matter how sparse $x_0$ is. This phenomenon is also characterized in \cite{DoT08}. 

The property that $\{x ~|~ Ax= Ax_0, x \geq 0\}$ is a singleton can
also be characterized in both high-dimensional geometry \cite{DoT08}
and the null
space property of $A$ \cite{KDXH09}. We state two necessary and sufficient conditions 
in Theorem \ref{thm:eqv}.

\begin{theorem} [\cite{DoT08,KDXH09}]
\label{thm:eqv} 
The following three properties of $A^{m\times n}$ are equivalent:
\begin{itemize}
\item For any nonnegative vector $x_0$ with a support size no greater than $k$, the set $\{x ~|~ Ax= Ax_0, x \geq
0\}$ is a singleton.
\item The polytope $P$ defined in (\ref{eqn:polytope}) has $n$ vertices and is $k$-neighborly.
\item For any $w \neq 0$ in the null space of $A$, both the positive support and the negative support of $w$ have a size of at least $k+1$.
\end{itemize}
\end{theorem}

Note that a polytope $P$ is $k$-neighborly if every set of $k$
vertices spans a face $F$ of $P$. $F$ is a face of $P$ if there
exists a vector $\alpha_F$ such that $\alpha_F^T x=c, \forall x \in
F$, and $\alpha_F^T x<c, \forall x \notin F$ and $x \in P$.

\cite{DoT08} (Corollary 4.1) shows that there exists a 
special partial Fourier matrix $\Omega$ with $2p+1$ rows 
such that $\{x ~|~ \Omega x= \Omega x_0, x \geq 0\}$ is a singleton
for every nonnegative $p$-sparse signal $x_0$. Here we will show the
result is the ``best'' we can hope for in the sense that a matrix
$A$ should have at least $2p+1$ rows if $\{x ~|~ Ax= Ax_0, x \geq
0\}$ is a singleton for every nonnegative $p$-sparse signal $x_0$.

\begin{prop}
For a matrix $A^{m\times n}$ $(m<n)$, if $\{x ~|~ Ax= Ax_0, x \geq
0\}$ is a singleton for any nonnegative $p$-sparse signal $x_0$,
then $m \geq 2p+1$.
\end{prop}
\begin{proof}
Pick the first $m+1$ columns of $A$, denoted by $a_1, a_2,...,
a_{m+1} \in \mathbb{R}^m$. Then the equations
\begin{equation}\label{eqn:best}
\sum_{i=1}^{m+1} \lambda_i a_i=0
\end{equation}
have $m$ equations and $m+1$ variables $\lambda_1,
\lambda_2,...,\lambda_{m+1}$, and have a non-zero solution.

From Theorem \ref{thm:general} we know that $A$ must belong to
$\mathbf{M}^+$, i.e. there exists $h$ such that $h^T A =\beta^T
>0$. Taking the inner product of both sides of (\ref{eqn:best}) with $h$,
we have
\begin{equation}\label{eqn:bestlambda}
\sum_{i=1}^{m+1}\beta_i \lambda_i=0.
\end{equation}

Since $\beta >0$, from (\ref{eqn:bestlambda}) we know $\lambda$
should have both positive and negative terms. Collecting positive
and negative terms of $\lambda$ separatively, we can rewrite
(\ref{eqn:best}) as follows,
\begin{equation}\label{eqn:reorder}
\sum_{i \in I_p} \lambda_i a_i= -\sum_{i \in I_n} \lambda_i a_i,
\end{equation}
where $I_p$ is the set of indices of positive terms of $\lambda$ and
$I_n$ is the set of indices of negative terms. Note that
$|I_p|+|I_n| \leq m+1$. We also have $\sum_{i \in I_p} \lambda_i= -
\sum_{i \in I_n} \lambda_i \triangleq r >0$ from
(\ref{eqn:bestlambda}).

Suppose $m \leq 2p$, then $|I_p|+|I_n|\leq m+1 \leq 2p+1$, we assume
without loss of generality that $|I_p| \leq p$. Since $\{x ~|~ Ax=
Ax_0, x \geq 0\}$ is a singleton for every nonnegative $p$-sparse
signal $x_0$, then from Theorem \ref{thm:eqv} 
$\textbf{Conv}(a_1,a_2,...,a_n)$ is $p$-neighborly, which implies
that for any index set $I$ with $|I|=p$, there exists $\eta$ such
that $\eta^T a_i =c$ for any $i \in I$, and $\eta^T a_i<c$ for all
$i \notin I$. We consider specifically an index set $I$, which
contains $I_p$ but does not contain $I_n$, and its corresponding
vector $\eta$. Taking the inner product of both sides of
(\ref{eqn:reorder}) with $\eta$, we would get $rc$ on the left and
some value strictly greater than $rc$ on the right, and reach a
contradiction.
\end{proof}

Sparse recovery problems appear in different fields. Specific
problem setup may impose further constraints on the measurement
matrix. We are particularly interested in network inference
problems, in which the measurement matrix is a 0-1 routing matrix.
Network inference problems attempt to extract individual parameters
based on aggregate measurements in networks. 
There has been active research in this area including a wide
spectrum of approaches ranging from theoretical reasoning to
empirical measurements \cite{CHNY02,NgT06,Duf06,ZRLD03,NgT07}. 
Since the measurement matrices in network inference problems are 0-1
matrices, the instances when $A$ is a 0-1 matrix are our main focus.
Section \ref{sec:bernoulli} and \ref{sec:binary} prove that a sparse
vector can be the unique nonnegative vector satisfying compressed
linear measurements if the measurement matrix is a random Bernoulli
matrix or an adjacency matrix of an expander graph. 
Moreover, the support size of the sparse vector can be proportional
to the dimension, in other words, the support size of the unique
nonnegative vector is $O(n)$ where $n$ is the dimension, while the
provable support size for uniqueness property in \cite{BEZ08} is
$O(\sqrt {n})$. Besides, for any $\theta \triangleq \lim \limits_{n
\rightarrow +\infty} \frac{m}{n} >0$, the support size of a sparse
vector that is a unique nonnegative solution can always be $O(n)$,
while for Gaussian measurement matrices, with high probability, $\{x
~|~ Ax= Ax_0, x \geq 0\}$ would not be a singleton for any
nonnegative $x_0$ (with linearly growing sparsity) if
$\theta<\frac{1}{2}$ \cite{DoT08}. This also shows the fundamental
difference between 0-1 measurement matrices and well studied
Gaussian random measurement matrices.

\subsection{Uniqueness with 0-1 Bernoulli
Matrices}\label{sec:bernoulli}

First we consider the uniqueness property with dense 0-1 Bernoulli
matrix. The measurement matrix $A$ is an $(m+1) \times n$
measurement matrix, with each element in the first $m$ rows of $A$
being i.i.d. Bernoulli random variables, taking values `0' with
probability $\frac{1}{2}$ and taking values `1' with probability
$\frac{1}{2}$. The last row of $A$ is a $1\times n$ all `1' vector.
We also assume the fraction ratio $\frac{m}{n}$ is a constant
$\theta$ as the dimension $n$ grows. It turns out that as $n$ goes
to infinity, with overwhelming probability there exists a constant
$\gamma>0$ such that $\{x ~|~ Ax= Ax_0, x \geq 0\}$ is a singleton
for any nonnegative $(\gamma n-1)$-sparse signal $x_0$. To see this,
we first present the following theorem:

\begin{theorem}\label{thm:bernoulli}
For any $\theta>0$, there exists a constant $\gamma >0$ such that,
with overwhelmingly high probability as $n \rightarrow \infty$, any
nonzero vector $w$ in the null space of the measurement $A$
mentioned above has at least $\gamma n$ negative and at least $\gamma n$ positive
elements.
\end{theorem}

\begin{proof}
Let us consider an arbitrary nonzero vector $w$ in the null space of
$A$. Let $S$ be the support set for the negative elements of $w$ and
let $S^c$ be the support set for the nonnegative elements of $w$. We
now want to argue that, with overwhelmingly high probability, the
cardinality $|S|$ of the set $S$ can not be too small.

From the large deviation principle and a simple union bound, for any
$\epsilon>0$, with overwhelmingly high probability as $n$ goes to
infinity, \emph{simultaneously} for \emph{every} column of the
measurement matrix, the sum of its $(m+1)$ elements will be in the
range $[\frac{1}{2}\theta(1-\epsilon)n,
\frac{1}{2}\theta(1+\epsilon)n]$.

Since $Aw=0$,
\begin{equation*}
A_{S}w_{S}+A_{S^c}w_{S^c}=0,
\end{equation*}
where $A_{S}$, $w_{S}$, $A_{S^c}$, and $w_{S^c}$ are respectively
the part of matrix $A$ and vector $w$ indexed by the sets $S$ and
$S^c$.

Multiplying the $1 \times m$ row vector $[1,1, ..., 1]$ to both
sides of this equation, we get
\begin{equation}
U_{S}w_{S}+U_{S^c}w_{S^c}=0, \label{eq:sumzeroconstraint}
\end{equation}
where $U_{S}$ is an $1 \times |S|$ vector, each component of which
represents the sum of the elements from the corresponding column of
$A_S$; $U_{S^c}$ is an $1 \times |S^c|$ vector, each component of
which represents the sum of the elements from the corresponding
column of $A_{S^c}$.

From the concentration result of the column sums, we know

\begin{equation*}
U_{S}w_{S} \geq  -\frac{1}{2}\theta(1+\epsilon)n \|w_{S}\|_1,
\end{equation*}
and
\begin{equation*}
U_{S^c}w_{S^c} \geq  \frac{1}{2}\theta(1-\epsilon)n \|w_{S^c}\|_1.
\end{equation*}

But combining these two inequalities with
(\ref{eq:sumzeroconstraint}), it follows that

\begin{equation*}
\frac{1}{2}\theta(1-\epsilon)n
\|w_{S^c}\|_1-\frac{1}{2}\theta(1+\epsilon)n \|w_{S}\|_1 \leq 0,
\end{equation*}
which implies

\begin{equation}
\frac{\|w_{S}\|_1}{\|w_{S^c}\|_1} \geq
\frac{1-\epsilon}{1+\epsilon}. \label{eq:ratioconstraint}
\end{equation}

Now we look at the null space of the measurement matrix $A$. First,
notice that the null space of $A$ is a subset of the null space of
the matrix $A'$ comprising of the first $\theta n$ rows of $A$
subtracted by the last row of $A$ (the all `$1$' vector). Then the
matrix $A'$ is a random $\pm 1$ Bernoulli measurement matrix, which
is known to satisfy the restricted isometry condition. Recall one
result about the null space property of a matrix satisfying the
restricted isometry condition:
\begin{lemma}[\cite{Can08}]\label{lem:Can}
Let $h$ be any vector in the null space of $A'$ and let $T_{0}$ be
any set of cardinality $q$. Then

\begin{equation*}
\|h_{T_{0}}\|_1 \leq \frac{\sqrt{2}\delta_{2q}}{1-\delta_{2q}}
\|h_{T_{0}^{c}}\|_1,
\end{equation*}
where $\delta_{2q}$ is the restricted isometry constant for sparse
vectors with support set size no bigger than $2q$, namely,
$\delta_{2q}$ is the smallest positive number such that for any set
$T$ with $|T| \leq 2q$, 
and any vector $y$, 
the following holds:
\begin{equation*}
 \sqrt{m}(1-\delta_{2q}) \|y\|_2 \leq \|A'_{T} y\|_2 \leq \sqrt{m} (1+\delta_{2q})
\|y\|_2.
\end{equation*}

\end{lemma}
Reasoning from Lemma \ref{lem:Can} and (\ref{eq:ratioconstraint}),
after some algebra, we know immediately, for $q=|S|$, $\delta_{2q}$
must satisfy

\begin{equation*}
\delta_{2q} \geq \frac{1-\epsilon}{1-\epsilon+\sqrt{2}(1+\epsilon)}.
\end{equation*}

We also know there exists a $\gamma>0$ such that for any $q\leq
\gamma n$, with overwhelmingly high probability as $n \rightarrow
\infty$,
\begin{equation*}
\delta_{2q} < \frac{1-\epsilon}{1-\epsilon+\sqrt{2}(1+\epsilon)},
\end{equation*}
thus with overwhelmingly high probability as $n \rightarrow \infty$,
the size of the negative support, namely $|S|$, can not be smaller
than $\gamma n$.

Similarly, we have the same conclusion for the cardinality of the
support set of the positive elements for any nonzero vector from the
null space of the matrix $A$.
\end{proof}


Theorem \ref{thm:bernoulli} immediately indicates that $\{x ~|~ Ax=
Ax_0, x \geq 0\}$ is a singleton for all nonnegative $x_0$ that is
$\gamma n-1$ sparse. Thus the support size of the unique nonnegative
vector can be as large as $O(n)$, while the previous result in
\cite{BEZ08} is $O(\sqrt{n})$.

\subsection{Uniqueness with Expander Adjacency Matrices}
\label{sec:binary}

\begin{figure}[t]
      \centering
      \includegraphics[scale=0.4]{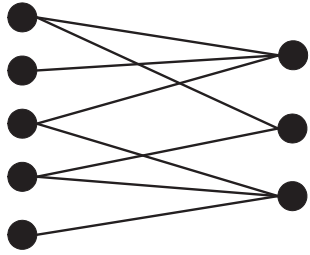}
 \caption{The bipartite graph corresponding to matrix $A$ in (\ref{eqn:routing_example})}
   \label{fig:bipartite}
\end{figure}

Section \ref{sec:bernoulli} discusses the singleton property with
0-1 Bernoulli matrices, here we focus on another type of 0-1
matrices where the matrix $A$ 
is the adjacency matrix of a bipartite expander graph. 
\cite{BeI08,XuH07,KDXH09} studied related problems using expander
graph with constant left degree. We instead employ a general
definition of expander which does not require constant left degree.

Every $m\times n$ binary matrix $A$ is the adjacency matrix of an
unbalanced bipartite graph with $n$ left nodes and $m$ right nodes.
There is an edge between right node $i$ and left node $j$ if and
only if $A_{ij}=1$. Let $d_j$ denote the degree of left node $j$,
and let $d_l$ and $d_u$ be the minimum and maximum of left degrees.
Define $\rho= d_l/d_u$, then $0< \rho \leq 1$. For example, the
bipartite graph in Fig. \ref{fig:bipartite} corresponds to the
matrix $A$ in (\ref{eqn:routing_example}). Here $d_l=1$, $d_u=2$,
and $\rho=0.5$.
\begin{equation}
A  = \left[ {\begin{array}{*{10}c}
   1 & 1 & 1 & 0 &  0   \\
   1 & 0 & 0 & 1 &  0   \\
   0 & 0 & 1 & 1 &  1
\end{array}} \right].
\label{eqn:routing_example}
\end{equation}
\begin{defi}[\cite{LMSS98}] A bipartite graph with $n$ left nodes and $m$ right nodes is an $(\alpha, \delta)$ expander if for any
set $S$ of left nodes of size at most $\alpha n$, the size of the
set of its neighbors $\Gamma(S)$ satisfies $|\Gamma(S)| \geq \delta
|E(S)|$, where $E(S)$ is the set of edges connected to nodes in $S$,
and $\Gamma(S)$ is the set of right nodes connected to $S$.
\end{defi}

Our next main result regarding the singleton property of an
adjacency matrix of a general expander is stated as follows.

\begin{theorem}\label{thm:expander} For an adjacency matrix $A$ of an $(\alpha, \delta)$ expander with left degrees in the range $[d_l, d_u]$, if $\delta \rho  > \frac{\sqrt{5}-1}{2}\approx0.618$, then for any nonnegative $k$-sparse vector $x_0$
with $k \leq \frac{\alpha}{1+\delta \rho}n$, $\{x ~|~ Ax=Ax_0, x\geq
0\}$ is a singleton.
\end{theorem}

\begin{proof}
From Theorem \ref{thm:eqv}, in order to prove that $\{x ~|~Ax=Ax_0,
x\geq 0\}$ is a singleton for any nonnegative
$\frac{\alpha}{1+\delta \rho}n$-sparse vector $x_0$, we only need to
argue that for any nonzero $w$ such that $Aw=0$, we have $|S_-| \geq
\frac{ \alpha n}{1+\delta \rho}+1$ and $|S_+| \geq \frac{ \alpha
n}{1+\delta \rho}+1$, where $S_-$ and $S_+$ are negative support and
positive support of $w$ respectively.

We will prove by contradiction. Suppose without loss of generality
that there exists a nonzero $w$ in $\mathbf{Null}(A)$ such that $|S_-|
=s \leq \frac{
 \alpha n}{1+\delta \rho}$, then the set $E(S_-)$ of
edges connected to nodes in $S_-$ satisfies
\begin{equation}\nonumber
d_l s \leq |E(S_-)| \leq d_u s.
\end{equation}
Then the set $\Gamma(S_-)$ of neighbors of $S_-$ satisfies
\begin{equation}\label{eqn:gamma_wn}\nonumber
d_u s \geq |E(S_-)| \geq |\Gamma(S_-)| \geq \delta |E(S_-)| \geq
\delta d_l s,
\end{equation}
where the second to last equality comes from the expander property.

Notice that $\Gamma(S_-)=\Gamma(S_+)=\Gamma(S_- \cup S_+)$,
otherwise $Aw=0$ does not hold, then
\begin{equation}\nonumber
|S_+| \geq  \frac{|\Gamma(S_+)|}{d_u}= \frac{|\Gamma(S_-)|}{d_u}\geq
\frac{ \delta d_l s}{d_u}= \delta \rho s.
\end{equation}

%

Now consider the set $S_- \cup S_+$, we have $|S_- \cup S_+| \geq
(1+\delta \rho)s$. Pick an arbitrary subset $\tilde {S} \in S_- \cup
S_+$ such that $|\tilde {S}|=(1+\delta \rho)s \leq \alpha n$.
From expander property, we have
\begin{equation}\nonumber
 |\Gamma(\tilde {S})| \geq  \delta |E(\tilde {S})| \geq \delta d_l |\tilde {S}|
 =\delta \rho (1+\delta \rho)d_u s > d_u s.
\end{equation}
The last inequality holds since $\delta \rho (1+\delta \rho)>1$
provided $\delta \rho > \frac{\sqrt{5}-1}{2}$.  But $|\Gamma(\tilde
{S})| \leq |\Gamma (S_- \cup S_+)|=|\Gamma(S_-)| \leq d_u s$. A
contradiction arises, which completes the proof.
\end{proof}

\begin{cor}\label{cor:regular}  For an adjacency matrix $A$ of an $(\alpha, \delta)$ expander with constant left degree $d$, if $\delta > \frac{\sqrt{5}-1}{2}$, then for any nonnegative $k$-sparse vector $x_0$
with $k \leq \frac{\alpha}{1+\delta}n$, $\{x ~|~ Ax=Ax_0, x\geq 0\}$
is a singleton.
\end{cor}

Theorem \ref{thm:expander} together with Corollary \ref{cor:regular}
is an extension to existing results. Theorem 3.5 of \cite{KDXH09}
shows that for an $(\alpha, \delta)$ expander with constant left
degree $d$, if $d \delta >1$, then there exists a matrix $\tilde{A}$
(a perturbation of $A$) such that $\{x ~|~\tilde{A}x=\tilde{A}x_0,
x\geq
0\}$ is a singleton for every nonnegative 
$\delta \alpha n$-sparse $x_0$. Our result instead can directly
quantify the sparsity threshold needed for a vector to be a unique
solution to compressed measurements induced by $A$, not its
perturbation.
\cite{BeI08} discussed the success of $L_1$ recovery of a general
vector $x$ for expanders with constant left degree. If we apply
Theorem 1 of \cite{BeI08} to cases where $x$ is known to be
nonnegative, the result can be interpreted as that $\{x ~|~Ax=Ax_0,
x\geq 0\}$ is a singleton for any nonnegative $\frac{\alpha}{2}
n$-sparse vector $x_0$ 
if $\delta
> \frac{5}{6} \approx 0.833$. Our result in Corollary \ref{cor:regular} implies that if $\delta > \frac{\sqrt{5}-1}{2}\approx
0.618$, $x_0$ can be $\frac{\alpha}{1+\delta}  n$-sparse and still
be the unique nonnegative solution. 

\cite{SiS96,FMSSW07} proved that for any $m$, $n$ and $\delta>0$,
there exists an $(\alpha, \delta)$ expander with constant left
degree $d$ for some $d$ and $\alpha>0$, and such an expander can be
generated through random graphs. There also exist explicit
constructions of expander graphs \cite{CRVW02}. Combining the
results with Corollary \ref{cor:regular}, for any $m$ and $n$, we
can generate an $(\alpha, \delta)$ expander with adjacency matrix
$A$ such that $\{x ~|~ Ax= Ax_0, x \geq 0\}$ is a singleton for any
nonnegative $kn$-sparse $x_0$, where $k=\frac{\alpha}{1+\delta}>0$.
Thus, same as Bernoulli 0-1 matrices, the adjacency matrix $A$ of an
$(\alpha, \delta)$ expander has the property that $\{x ~|~ Ax= Ax_0,
x \geq 0\}$ is a singleton as long as the support size of $x_0$ is
$O(n)$. We further provide an explicit constant
$\frac{\alpha}{1+\delta}$ of the ratio of the support size to the
dimension. Note that this result is independent of the ratio
$\frac{m}{n}$, while as discussed earlier, if the matrix has i.i.d.
Gaussian entries and $\lim \limits_{n \rightarrow +\infty}
\frac{m}{n}<\frac{1}{2}$, $\{x ~|~ Ax= Ax_0, x \geq 0\}$ is not a
singleton despite the sparsity of $x_0$.

\section{Unique Positive Semidefinite Solution to an Underdetermined
System} \label{sec:matrix}
\subsection{When is Low-rank Positive Semidefinite Solution the
Unique Solution?} \label{sec:iff}
Section \ref{sec:vector} studies the case when a sparse nonnegative
vector is the only nonnegative solution to the system of compressed
linear measurements. Here we extend the problem into the matrix
space. Let $X$ be an $n \times n$ matrix decision variable. Let
$\mathcal{A}:\mathbb{R}^{n\times n}\rightarrow \R^m$ be a linear
map, and let $b \in \R^m$. The main optimization problem under study
for low-rank matrix recovery is
\begin{equation}
\begin{array}{ll}
\minimize & \rank(X)\\
\st & \cA(X)=b\,.
\end{array}
\label{eq:min-rank-prob}
\end{equation}

In this paper, we are interested in looking at the property of the
feasible set $\{X'~|~ \cA(X')=b\}$. 
Indeed, if there exists a $X'$ such that $\cA(X')=b$, then $X'$ plus
any matrix in the null space of $\cA$ also satisfies $\cA(X')=b$.
However, in applications, one is often interested in recovering a
positive semidefinite symmetric matrix $X$, ($X\succeq 0$ and $X\in
S^n$, where $S^n$ is the set of $n \times n$ real symmetric
matrices) from compressed observations. To determine a positive
semidefinite symmetric matrix $X$, we only need to determine
$\frac{n(n+1)}{2}$ unknowns in the upper triangular part of $X$.
Thus the linear operator $\cA$ in (\ref{eq:min-rank-prob}) can be
reduced to an operator
$\mathcal{A}(X^{\bot}):\mathbb{R}^{\frac{n(n+1)}{2}}\rightarrow
\R^m$, where $m \leq \frac{n(n+1)}{2}$ and $X^{\bot}$ denotes the
upper triangular part of the $n \times n$ symmetric matrix $X$. The
null space of $\mathcal{A}$ is a subset of
$\mathbb{R}^{\frac{n(n+1)}{2}}$ such that each point from this set,
arranged accordingly as the upper triangular part of $Y$ of a $n
\times n$ matrix $Y$, satisfies $\mathcal{A}(Y)=0 \in \mathbb{R}^m$.

Now we ask this question, can we uniquely determine the positive
semidefinite symmetric matrix $X$ from $\cA(X)=b$, namely can the
feasible set $\{X'~|~ \cA(X')=b, X'\succeq 0, X'\in S^n\}$ be a
singleton? The next theorem gives an affirmative answer to this
question, and shows that if the linear measurement operator
satisfies certain conditions and the positive semidefinite symmetric
matrix $X$ is of low rank, then the feasible set $\{X'~|~ \cA(X')=b,
X'\succeq 0, X'\in S^n\}$ is a singleton, namely $X$ is not only the
only low-rank solution, but also the only \emph{possible} solution.

\begin{theorem}\label{thm:iff} Let $X$ be a positive semidefinite symmetric matrix of rank $r$ and $\mathcal{A}:\mathbb{R}^{\frac{n(n+1)}{2}}\rightarrow \R^m$ be a linear operator which operates on the upper triangular part of $X$, where $m < \frac{n(n+1)}{2}$. Then $\{X'~|~ \cA(X')=\cA(X), X'\succeq 0, X'\in S^n\}$ is a singleton for all $X$ with rank no greater than $r$, if and only if
for every non-all-zero matrix generated from the null space of $\cA$
has at least $r+1$ negative eigenvalues.
\end{theorem}
\begin{IEEEproof}
Sufficiency: we first show that if every non-all-zero symmetric
matrix generated from the null space of $\cA$  has at least $r+1$
negative eigenvalues, then $\{X'~|~ \cA(X')=\cA(X), X'\succeq 0,
X'\in S^n\}$  is a singleton. Suppose instead there exist a $X''\in
S^n$ such that $\cA(X'')=b$, then the upper triangular part of
$X''-X$ is in the null space of the linear operator $\cA$. By the
assumption, we know that $X''-X$ has at least $r+1$ negative
eigenvalues. Since $X''-X$ is a symmetric matrix, its eigenvalues
are real. For a matrix, we denote these eigenvalues in an
nondecreasing order, namely,
\begin{equation*}
\lambda_1\leq\lambda_2\leq\cdots\lambda_{n-1}\leq \lambda_{n}.
\end{equation*}

 By a classical variational characterization of eigenvalues \cite{HornJohnsonBook1}, if $A$ and $B$ are both $n \times n$ Hermitian matrices and $B$ has rank at most $r$, then $\lambda_{k}(A+B) \leq \lambda_{k+r}(A)$, for $k=1,2,..., n-r$. By taking $k=1$, $B=X$ and $A=X''-X$, we have
\begin{equation*}
\lambda_{1}(X'')=\lambda_{1}((X''-X)+X)\leq \lambda_{r+1}(X''-X)<0,
\end{equation*}
by the eigenvalue assumption for $X''-X$. But then $X''$ is not a
positive semidefinite matrix. This contradiction shows that $X$ is
the only element in the the set $\{X'~|~ \cA(X')=\cA(X), X'\succeq
0, X'\in S^n\}$.

Necessity: we need to show that if there exists a nontrivial
symmetric matrix (say $Y$), with its upper triangular part from the
null space of the linear operator $\cA$, has at most $r$ negative
eigenvalues, then we can find an $X$ such that $\{X'~|~
\cA(X')=\cA(X), X'\succeq 0, X'\in S^n\}$ is not a singleton.
Indeed, since $Y$ is a symmetric matrix, it can be diagonalized by
some unitary matrix $U$, namely $Y=U \Lambda U^{-1}$, where
$\Lambda$ is a diagonal matrix with $\Lambda_{i,i}=\lambda_i(Y)$. We
then pick $X=U\Lambda'U^{-1}$, where $\Lambda'$ is a diagonal
matrix, and $\Lambda'_{i,i}>\max\{-\lambda_{i},0\}$ for $1\leq i\leq
r$ and $\Lambda'_{i,i}=0$ for $i>r$. Thus $X$ is a positive
semidefinite matrix with rank no larger than $r$ (note that the
eigenvalues of $\Lambda'$ are not necessarily arranged in
nondecreasing order with respect to $i$ ). Then obviously
$X+Y=U\Lambda''U^{-1}$, where the diagonal entries in the diagonal
matrix $\Lambda''=\Lambda'+\Lambda$ are all nonnegative. Since $Y$
is not a all-zero matrix, $X+Y$ is an element in the set $\{X'~|~
\cA(X')=\cA(X), X'\succeq 0, X'\in S^n\}$ besides $X$.

\end{IEEEproof}

Theorem \ref{thm:iff} establishes the necessary and sufficient
condition for the uniqueness of low-rank positive semidefinite
solution under compressed linear measurements. However, checking
this condition for a specific set of linear measurements seems to be
a hard problem and, in addition, it is not clear whether
asymptotically there exist such linear compressed measurements
satisfying the given condition. So in Section \ref{sec:null}, we
will investigate whether a set of linear measurements (namely the
linear measurement $\cA(\cdot)$) sampled from a certain distribution
will satisfy this condition.

\subsection{The Null Space Analysis of the Gaussian Ensemble}
\label{sec:null}

 We say that the linear operator $\mathcal{A}:\mathbb{R}^{\frac{n(n+1)}{2}}\rightarrow \R^m$ is sampled from an independent Gaussian ensemble if its $i$-th ($1\leq i\leq m$) operation, denoted by $\mathcal{A}_{i}:\mathbb{R}^{\frac{n(n+1)}{2}}\rightarrow \R$, is the inner product
 \begin{equation*}
 \langle X, A_i\rangle= trace(X^TA_i),
 \end{equation*}
 where $A_i$ is an $n \times n$ symmetric matrix with independent random elements in its upper triangular part. On the diagonal of $A_i$, its elements are distributed as real Gaussian random variables $N(0,1)$ and, in the off-diagonal part, its elements are distributed as $N(0,\frac{1}{2})$. Across the index $i$, the $A_i$'s are also sampled independently. One main result of this paper can be stated in the following theorem.

\begin{theorem}\label{thm:existence2} Consider a linear operator $\mathcal{A}:\mathbb{R}^{\frac{n(n+1)}{2}}\rightarrow \R^m$ sampled from an independent Gaussian ensemble. Let $m=\alpha \times \frac{n(n+1)}{2}$. Then there exists a constant $\alpha<1$, independent of $n$, such that with overwhelming probability as $n$ goes to $\infty$, any nonzero symmetric $n \times n$ square matrix with its upper triangular part from the null space of the linear operator $\cA$ has at least $\xi n$ negative eigenvalues, where $\xi>0$ is a constant that is independent of $n$. Thus with overwhelmingly high probability,  any positive semidefinite matrix of rank no larger than $\xi n-1$ will be the singleton in the set $\{X'~|~ \cA(X')=\cA(X), X'\succeq 0, X'\in S^n\}$.
\end{theorem}

Note that in Theorem \ref{thm:existence2}, the constant $\xi$ may
depend on $\alpha$. Theorem \ref{thm:existence2} confirms that there
indeed exists a sequence of linear operators such that \emph{every}
nonzero element in their null spaces necessarily generates a
symmetric matrix having a sufficiently large number ($\xi n$) of
negative eigenvalues. The ``guaranteed'' number of negative
eigenvalues is highly nontrivial in the sense that $\xi n$ grows
proportionally with $n$ while the null space for the linear operator
$\cA$ has dimension at least $(1-\alpha)\frac{n(n+1)}{2}$, which
grows proportionally with $n^2$. This seems counterintuitive at
first sight: a null space of such a large dimension should have been
able to accommodate at least one point which generates a symmetric
matrix with very few or even none negative eigenvalues.

The main difficulty in proving Theorem \ref{thm:existence2} is to
show that for \emph{all} the nonzero symmetric matrices generated
from the points in the null space of the random linear operator
$\cA$, the claimed fact holds \emph{universally} with overwhelming
probability. This seems to be a daunting job since the null space of
every linear operator is a continuous object and there are
uncountably many symmetric matrices that can be generated from it.
In fact, we have the following probabilistic characterization with a
shortened proof for the null space of the linear operator sampled
from the independent Gaussian Ensemble.
\begin{lemma}
If the linear operator
$\mathcal{A}(X):\mathbb{R}^{\frac{n(n+1)}{2}}\rightarrow \R^m$ is
sampled from independent Gaussian Ensemble, by representing the vectors from the null
space of $\mathcal{A}$ by $\frac{n(n+1)}{2} \times 1$ column
vectors, the distribution of its null space is (almost everywhere)
equivalent to the distribution of a
$(\frac{n(n+1)}{2}-m)$-dimensional subspace in
$\mathbb{R}^{\frac{n(n+1)}{2}}$ whose basis can be represented by a
$\frac{n(n+1)}{2} \times (\frac{n(n+1)}{2}-m)$ matrix $Z$ whose
elements are independent Gaussian random variables, $N(0,1)$ for
elements in the rows corresponding to the $n$ diagonal elements of
$X$ and $N(0,\frac{1}{2})$ for elements in the rows corresponding to
the $\frac{n(n-1)}{2}$ off-diagonal elements. \label{lem:nulldist}
\end{lemma}

\begin{proof}
This lemma follows from the fact that a random matrix with zero mean
i.i.d. Gaussian distributed entries generates a random subspace
whose distribution is rotationally invariant (namely the
distribution of that random subspace does not change when it is
rotated by a unitary rotation). We also note that if a random
subspace has a rotationally invariant distribution, its null space
also has a rotationally invariant distribution, which again can be
generated by a matrix with zero mean i.i.d. Gaussian distributed
entries of appropriate dimensions (with probability 1, the dimension
of this null space is $(\frac{n(n+1)}{2}-m)$). With a normalization
for the variance of the Gaussian distributed entries, we have this
lemma.
\end{proof}

By Lemma \ref{lem:nulldist}, the null space of the linear operator
$\cA$ sampled from independent Gaussian Ensemble can be represented
by

\begin{equation}\nonumber
\{z~|~z=Zw, w\in \mathbb{R}^{\frac{n(n+1)}{2}-m}\},
\end{equation}
where $Z$ is a $\frac{n(n+1)}{2} \times (\frac{n(n+1)}{2}-m)$ matrix
as mentioned in Lemma \ref{lem:nulldist}.

We should first notice that in order to prove the property that
``any nonzero symmetric $n \times n$ square matrix with its upper
triangular part from the null space of the linear operator $\cA$ has
at least $\xi n$ negative eigenvalue'' , we only need to restrict
our attention to prove that property for the set of symmetric
matrices generated by the set of points
\begin{equation}\nonumber
\{z~|~z=\frac{1}{\sqrt{n}}Zw, w\in \mathbb{R}^{\frac{n(n+1)}{2}-m}, \|w\|_2=1\},
\end{equation}
in the null space of the linear operator $\cA$.

Building on this observation, we can proceed to divide the formal
proof of Theorem \ref{thm:existence2} into three steps. Firstly,
since we can not show directly our theorem for every point in the
null space, instead we first try to discretize the sphere
\begin{equation}\nonumber
\{w~|~\|w\|_2=1, w\in \mathbb{R}^{\frac{n(n+1)}{2}-m}\}
\end{equation}
into a finite $\epsilon$-net consisting of a finite number of points
on the sphere  such that every point in the set $\{w~|~\|w\|_2=1,
w\in \mathbb{R}^{\frac{n(n+1)}{2}-m}\}$ is in the $\epsilon$ (in
terms of Euclidean distance) neighborhood of at least one point from
the $\epsilon$-net. Formally, an $\epsilon$-net is a subset
$\mathcal{S} \subset \{w~|~\|w\|_2=1, w\in
\mathbb{R}^{\frac{n(n+1)}{2}-m}\}$ such that for every point $t$ in
the set $\{w~|~\|w\|_2=1, w\in \mathbb{R}^{\frac{n(n+1)}{2}-m}\}$,
one can find $s$ in $\mathcal{S}$ such that $\|t-s\|_2\leq
\epsilon$. The following lemma is well known in high dimensional
geometry about the size estimate of such a $\epsilon$-net, for
example, see \cite{Ledoux00}:

\begin{lemma}
There is an  $\epsilon$-net $\mathcal{S}$ of the unit sphere of
$\mathbb{R}^{\frac{n(n+1)}{2}-m}$ of cardinality less than
$(1+\frac{2}{\epsilon})^{\frac{n(n+1)}{2}-m}$, which is no larger
than $e^{\frac{n(n+1)-2m}{\epsilon}}$.

\end{lemma}

Secondly, using the large deviation technique or concentration of
measure result, we establish the relevant properties for the
symmetric matrices generated from these discrete points on the
$\epsilon$-net. For example, the symmetric matrices have a large
number of negative eigenvalues with overwhelming probability.
Thirdly, we show how property guarantees on the $\epsilon$-net can
be used to establish the null space property for the whole null
space of the linear operator $\cA$. Section \ref{sec:single} and
\ref{sec:netanalysis} are then devoted to completing these steps to
prove Theorem \ref{thm:existence2}.

\subsection{Concentration for a Single Point} \label{sec:single} We
take any point $w$ from the $\epsilon$-net for the set
$\{w~|~\|w\|_2=1, w\in \mathbb{R}^{\frac{n(n+1)}{2}-m}\}$ and its
corresponding point $z=\frac{1}{\sqrt{n}}Zw$ in the null space of
the linear operator $\cA$, where $Z$ is the random basis as
mentioned in Lemma \ref{lem:nulldist}.  Then we argue that the
symmetric matrix $G$ with its upper triangular part generated from
$z$ has many negative eigenvalues with overwhelming probability. It
is obvious that with the i.i.d. Gaussian probabilistic model for
$Z$, the elements of $G$ are independently Gaussian distributed
$N(0,\frac{1}{n})$ random variables on the diagonal and
independently Gaussian distributed $N(0,\frac{1}{2n})$ on the
off-diagonal.

\begin{theorem}
 The smallest $\alpha_1 n$ ($\alpha_1<\frac{1}{2}$) eigenvalues of the symmetric matrix $G$ with its upper triangular part generated from $z$ will be upper bounded by $c+\delta$ with overwhelming probability $1-e^{-c_1n^2}$, where $c$ is a negative number as determined from the semicircular law
 \begin{equation*}
 \alpha_1=\frac{1}{\pi}\int^{c}_{-\infty}\mathbf{1}_{|x|<\sqrt{2}} \sqrt{2-x^2} \,dx,
 \end{equation*}
$\delta$ is an arbitrarily small positive number, $c_1$ is a
positive constant independent of $n$ and $\mathbf{1}$ is the
indicator function. \label{thm:eigenconcen}
\end{theorem}

\begin{proof}
Indeed Theorem \ref{thm:eigenconcen} can be derived from known large
deviations or concentration of measure results for the empirical
eigenvalue distribution of random symmetric Gaussian matrix
\cite{Alice97} \cite{Guionnet00}. Obviously, $G$ has $n$ real
eigenvalues ${(\lambda_i)}_{1\leq i\leq n}$ arranged in
nondecreasing order and its spectral measure
${\hat{\mu}}^{n}\triangleq\frac{1}{n}\sum_{i=1}^{n}\delta_{\lambda_i}=\frac{1}{n}\sum_{i=1}^{n}\delta(\lambda-\lambda_i)$,
where $\delta(\cdot)$ is the delta function. As in \cite{Alice97},
we denote the space of probability measure on $\mathbb{R}$ as
$\mathcal{M}_{1}^{+}(\mathbb{R})$ and will endow
$\mathcal{M}_{1}^{+}(\mathbb{R})$ with its usual weak topology.
\cite{Alice97} then gives the following large deviation result for
the empirical eigenvalue distribution for the matrix $G$,

\begin{theorem} [\cite{Alice97}]
Let $\mu \in \mathcal{M}_{1}^{+}(\mathbb{R})$, define the rate
function
\begin{equation}\nonumber
I_1(\mu)=\frac{1}{2}(\int{x^2}\,d\mu(x)-\Sigma(\mu))-\frac{3}{8}-\frac{1}{4}\log(2),
\end{equation}
where $\Sigma{(\mu)}$ is the non commutative entropy
\begin{equation}\nonumber
\Sigma{(\mu)}=\int{\int{\log(|x-y|)}\,d\mu(x)} \,d\mu(y).
\end{equation}

Then

\begin{itemize}
\item
  \begin{itemize}
  \item $I_1$ is well defined over the set $\mathcal{M}_{1}^{+}(\mathbb{R})$ and takes its value in $[0,+\infty)$;
  \item
    $I_1(\mu)$ is infinite as long as $\mu$ satisfies the
    following:
    \begin{itemize}
    \item $\int{x^2}\,dx=+\infty$
    \item there exists a subset $A$ of $R$ with a positive $\mu$ mass but null logarithmic capacity, i.e. a set $A$ such that $\mu(A)>0$ and
        \begin{equation*}
      \hspace{-.65in}  \noindent \gamma(A)=\exp\{-\inf_{\nu \in \mathcal{M}_{1}^{+}(\mathbb{R})}\int{\int{\log(\frac{1}{|x-y|})}\,d\nu(x)} \,d\nu(y) \}=0
        \end{equation*}
    \end{itemize}

    \item $I_1{(\mu)}$ is a good rate function, namely $\{I_1{(\mu)}\leq M\}$ is a compact subset of $\mathcal{M}_{1}^{+}(\mathbb{R})$ for $M\geq 0$.

    \item $I_1$ is a convex function on $\mathcal{M}_{1}^{+}(\mathbb{R})$.
    \item $I_1$ achieves its minimum value at a unique probability measure on $\mathbb{R}$ which is described by the Wigner's Semicircle Law.
 \end{itemize}

\item The law of the spectral measure ${\hat{\mu}}^{n}=\frac{1}{n}\sum_{i=1}^{n}\delta_{\lambda_i} $ satisfies a full large deviation principle with good rate function $I_1$ and in the scales $n^2$, that is,
    for any open subset $O$ of
    $\mathcal{M}_{1}^{+}(\mathbb{R})$,
\begin{equation*}
\liminf_{n \rightarrow \infty}\frac{1}{n^2} \log(P({\hat{\mu}}^{n} \in O)) \geq -\inf_{O}I_1
\end{equation*}
    for any closed subset $F$ of of
    $\mathcal{M}_{1}^{+}(\mathbb{R})$,
\begin{equation*}
\limsup_{n \rightarrow \infty}\frac{1}{n^2} \log(P({\hat{\mu}}^{n} \in F)) \leq -\inf_{F}I_1
\end{equation*}
\end{itemize}
\label{them:AliceLDP}
\end{theorem}

We take $c$ as in the statement of Theorem \ref{thm:eigenconcen} and
then the set of spectral measure $A$ satisfying the statement of
Theorem \ref{thm:eigenconcen} can be denoted by
$\{\frac{1}{n}\sum_{i=1}^{n} \mathbf{1}_{\lambda_i \leq
c+\delta}>\alpha_1\}$, whose complement is then
$\{\frac{1}{n}\sum_{i=1}^n \mathbf{1}_{\lambda_i\leq c+\delta}\leq
\alpha_1\}$.

Now we take a continuous function $f$ equal to $\mathbf{1}_{x\leq
c}$ over the region $(-\infty, c]$, equal to $0$ on
$[c+\delta,+\infty)$, and linear in between over the region
$[c,c+\delta]$. Then the set $A$ is included in the following set of
probability measure
\begin{eqnarray*}
\{\frac{1}{n}\sum_{i=1}^n f(\lambda_i)\leq \alpha_1\}=\{\hat{\mu}^n(f)\leq \alpha_1\}\subseteq \{\mu(f)\leq \alpha_1\}\triangleq B(\mu),
\end{eqnarray*}
with $\hat{\mu}^n=\frac{1}{n} \sum_{i=1}^{n}{\delta_{\lambda_i}}$
and $\mu(f)$ as the integral of $f$ over $\mu$.

This set $B$ is closed for the weak topology and so we can apply the
large deviation principle as in \cite{Alice97}. To get that
\begin{equation*}
\limsup_{n\rightarrow\infty}\frac{1}{n^2}\log P(\{\frac{1}{n}\sum_{i=1}^n
\mathbf{1}_{\lambda_i\leq c+\delta}\leq \alpha_1\})
\leq -\inf_{B}I
\end{equation*}
with $I$ as defined in Theorem \ref{them:AliceLDP}, from the
definition of $\alpha_1$, we simply know that the semi-circle law
does not belong to the set $B$ and so we can conclude that
$\inf_{B}I>0$. This is because the rate function $I$ is a good
rate function which achieves its unique minimum at the semicircle
law.

\end{proof}

Following Theorem \ref{thm:eigenconcen}, we know that with
overwhelming probability, the symmetric matrix generated from a
single point on the $\epsilon$-net will be very likely to have a
large number (proportional to $n$) of negative eigenvalues. In
Section \ref{sec:netanalysis}, we will show how to synthesize the
results for isolated points so that we can prove the eigenvalue
claim for the null space of the linear operator $\cA$.

\subsection{Concentration for the Null Space: $\epsilon$-net
Analysis} \label{sec:netanalysis} Building on the concentration
results for the single point on the $\epsilon$-net, we now begin
proving the claims in Theorem \ref{thm:existence2} for all the
possible symmetric matrices generated from the set
\begin{equation}\nonumber
\{z~|~z=Zw, w\in \mathbb{R}^{\frac{n(n+1)}{2}-m}\},
\label{set:randmbasis}
\end{equation}
where $Z$ is a $\frac{n(n+1)}{2} \times (\frac{n(n+1)}{2}-m)$ matrix
as mentioned in Lemma \ref{lem:nulldist}.

First, we make a simple observation regarding every point $w$ on the
Euclidean sphere
\begin{equation}\nonumber
\{w~|~\|w\|_2=1, w\in \mathbb{R}^{\frac{n(n+1)}{2}-m}\}.
\end{equation}

Since $\mathcal{S}$ is an $\epsilon$-net on the sphere, we can find
a point $w_{0} \in \mathcal{S}$ with $\|w_{0}\|_2=1$ such that
$\|w-w_{0}\|_2 \leq \epsilon$. For the error term $w-w_{0}$, we can
still find a point $w_1$ on the $\epsilon$-net $\mathcal{S}$ such
that
\begin{equation*}
\|w-w_0-\|w-w_0\|_2w_1\|_2 \leq \epsilon \|w-w_0\|_2 \leq \epsilon^2.
\end{equation*}
By iterating this process, we get that any $w$ on the unit Euclidean
sphere can be expressed as
\begin{equation}
w=w_0+\sum_{i=1}^{\infty}t_iw_i,
\label{eq:enetsum}
\end{equation}
where $|t_i|\leq \epsilon^i$ for $i \geq 1$ and $w_i \in
\mathcal{S}$ for $i\geq 0$.

Before we proceed further to look at the spectrum of the symmetric
matrix $B_w$ generated from $Zw$, we state the following theorem by
Hoffmann and Wielandt \cite{HornJohnsonBook1}.

\begin{theorem}[\cite{HornJohnsonBook1}]
Let $A$, $E \in M_{n}$, assume that $A$ and $A+E$ are both normal,
let $\lambda_1, ..., \lambda_n$ be the eigenvalues of $A$ in some
given order, and let $\hat{\lambda}_{1}, ..., \hat{\lambda}_{n}$ be
the eigenvalues of $A+E$ in some order. Then there exists a
permutation $\sigma_{i}$ of the integers $1$, $2$, ..., $n$ such
that
\begin{equation*}
\left[\sum_{i=1}^{n}|\hat{\lambda}_{\sigma_{i}}-\lambda_i|^2\right]^{\frac{1}{2}} \leq \|E\|_2
\end{equation*}
\label{thm:HoffmanWielandt}
\end{theorem}

Now we can give a closer study of the $n \times n$ symmetric matrix
$B_w$ generated from $\frac{1}{\sqrt{n}}Zw$. From the $\epsilon$-net
decomposition (\ref{eq:enetsum}), it follows that
\begin{equation}\nonumber
B_w=B_{w_0}+\sum_{i=1}^{n}{t_iB_{w_i}},
\end{equation}
where $B_{w_i}$ is the symmetric matrix generated from
$\frac{1}{\sqrt{n}}Zw_i$ for $i \geq 0$.

Since we can thus view $B_w$ as $B_{w_0}$ plus some perturbation,
using Theorem \ref{thm:HoffmanWielandt}, there exists a permutation
$\sigma_{i}$ of the integers $1$, $2$, ..., $n$ such that
\begin{equation}
\left[\sum_{i=1}^{n}|\hat{\lambda}_{\sigma_{i}}-\lambda_i|^2\right]^{\frac{1}{2}} \leq \|\sum_{i=1}^{n}{t_iB_{w_i}}\|_2,
\label{eq:hoff1}
\end{equation}
where $\hat{\lambda}_{i}$, $1\leq i \leq n$, and $\lambda_{i}$,
$1\leq i \leq n$, are the eigenvalues of the $B_{w}$  and
$B_{w_{0}}$ arranged in an increasing order, respectively.

But from the triangular inequality, we know
\begin{eqnarray}
\|\sum_{i=1}^{n}{t_iB_{w_i}}\|_2 &\leq& \sum_{i=1}^{n}|t_i|\|B_{w_i}\|_2\nonumber \\
&\leq& \sum_{i=1}^{n}\epsilon^i C_1\sqrt{n} \leq \frac{\epsilon C_1\sqrt{n}}{1-\epsilon},
\label{eq:hoff2}
\end{eqnarray}
where we use the fact (derivations omitted) that with overwhelmingly
high probability (the complement probability exponent in the scale
of $-n^2$) as $n \rightarrow \infty$, $\|B_{w}\|_2$ is upper bounded
by $C_1\sqrt{n}$ \emph{simultaneously} for all $w \in \mathcal{S}$
with $C_1$ as a constant independent of $n$.

Now we can officially argue that the number, say $k$, of negative
eigenvalues of $B_w$ can not be small. In particular, we will upper
bound $(\alpha_1 n-k)$, where $\alpha_1$ is as defined in Theorem
\ref{thm:eigenconcen} for $B_{w_0}$. By picking $c$ to be negative
and $\delta$ to be small enough in Theorem \ref{thm:eigenconcen},
$c+\delta$ will be negative. Then for whatever ordering of the
eigenvalues of $B_{w}$, we have
\begin{equation}
\left[\sum_{i=1}^{n}|\hat{\lambda}_{\sigma_{i}}-\lambda_i|^2\right] \geq (\alpha_1 n-k)|c+\delta|^2,
\label{eq:hoff3}
\end{equation}
because at least $(\alpha_1 n-k)$ negative eigenvalues (smaller than
$c+\delta$) of $B_{w_{0}}$ will be matched to nonnegative
eigenvalues of $B_{w}$ in Theorem \ref{thm:HoffmanWielandt}.

Connecting (\ref{eq:hoff1}), (\ref{eq:hoff2}) and (\ref{eq:hoff3}),
we have with overwhelming probability, simultaneously for every $w$
on the Euclidean sphere, if $(\alpha_1 n-k)\geq 0$, (otherwise $k$
already nicely bounded)
\begin{equation}\nonumber
 \sqrt{(\alpha_1 n-k)|c+\delta|^2} \leq \frac{\epsilon C_1\sqrt{n}}{1-\epsilon}.
\end{equation}
So
\begin{equation}\nonumber
k \geq \alpha_1 n-\frac{\epsilon^2 C_1^2 n}{(1-\epsilon)^2|c+\delta|^2},
\end{equation}
which implies if we pick $\epsilon$ small enough, the number of
negative eigenvalues of $B_{w}$ will be proportionally growing with
$n$.  Note that for any $\epsilon>0$,$c<0$, $\delta>0$ and $C_1>0$,
we can always find a large enough $\alpha=\frac{2m}{n(n+1)}$ to make
sure that the union bound exponent from the cardinality of the
$\epsilon$-net is overwhelmed by both the negative large deviation
exponent for the spectral measure and the negative large deviation
exponent for the Forbenius norm of the random matrix.  In summary,
we have arrived at a complete proof of Theorem \ref{thm:existence2}.

\section{Simulation}\label{sec:simulation}
\begin{figure*}
\centering
\begin{tabular}{c c}
\includegraphics[width=0.5\linewidth]{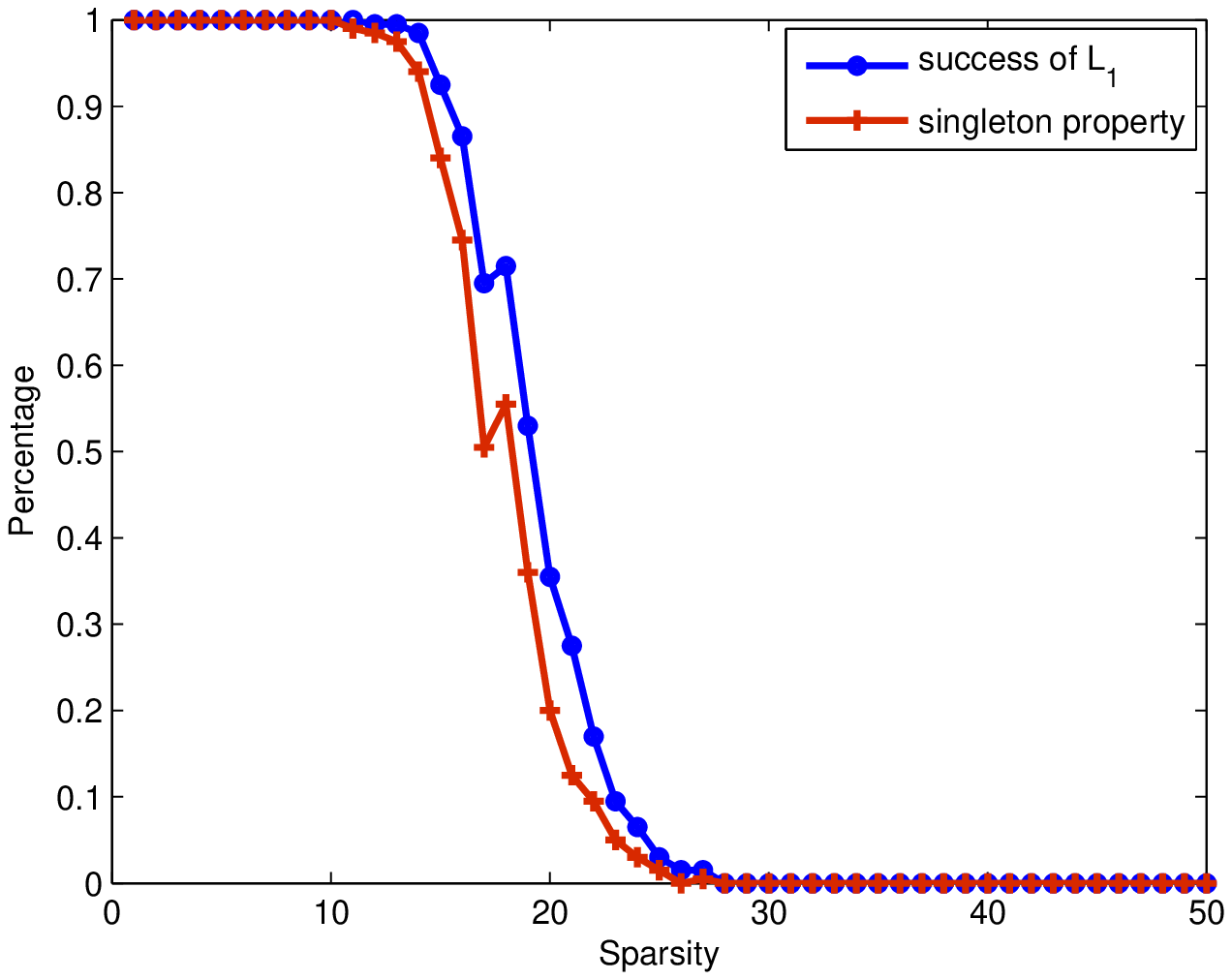}
&
\includegraphics[width=0.5\linewidth]{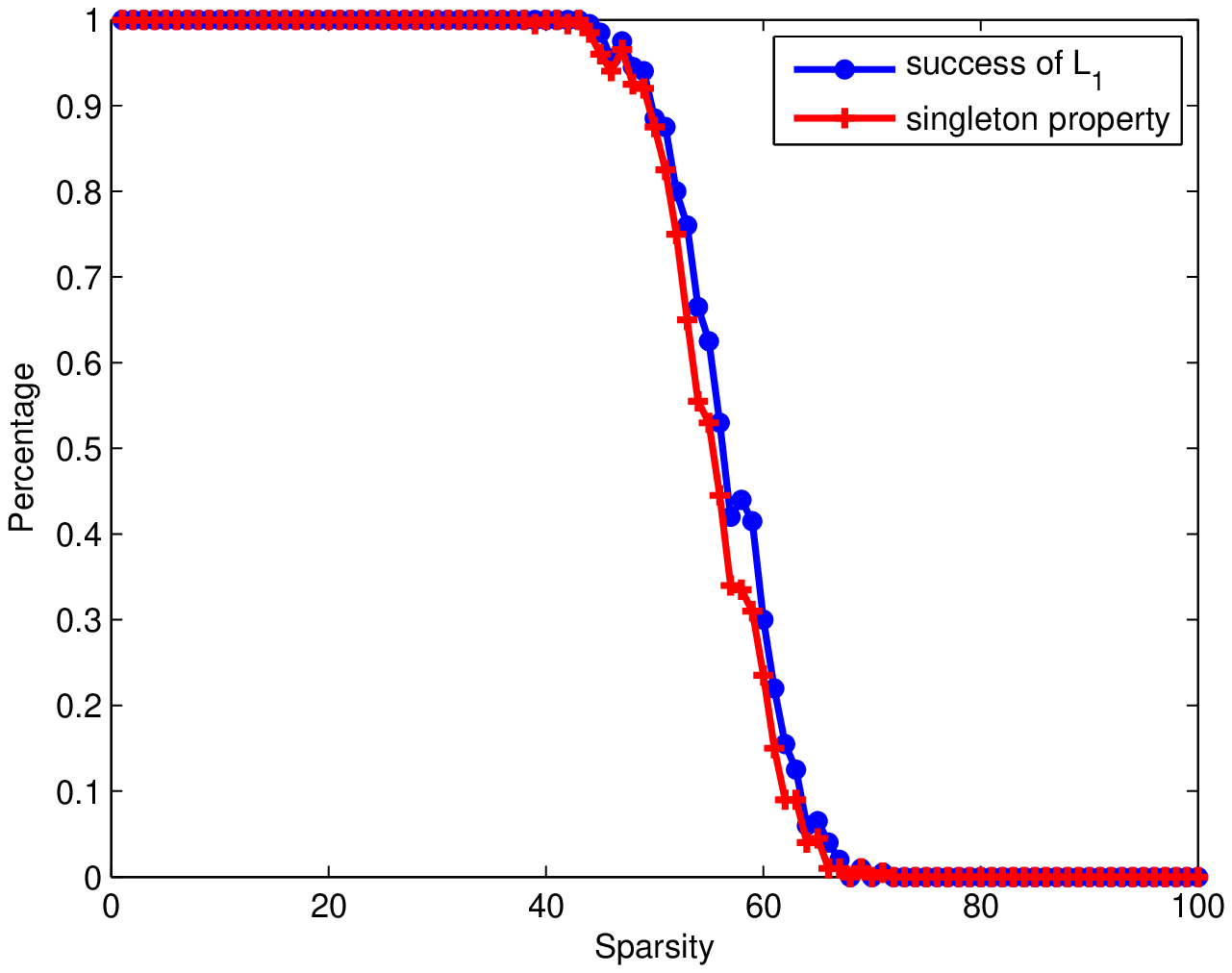}
\\
{\small (a) 50 $\times$ 200 0-1 matrix }& (b){\small 100 $\times$ 200 0-1 matrix}\\
\end{tabular}
      \caption{Comparison of $L_1$ recovery and singleton property for (a) 50 $\times$ 200 0-1 matrix and (b) 100 $\times$ 200 0-1 matrix }
      \label{fig:result}
   \end{figure*}

\begin{figure*}
\centering
\begin{tabular}{c c}
\includegraphics[width=0.5\linewidth]{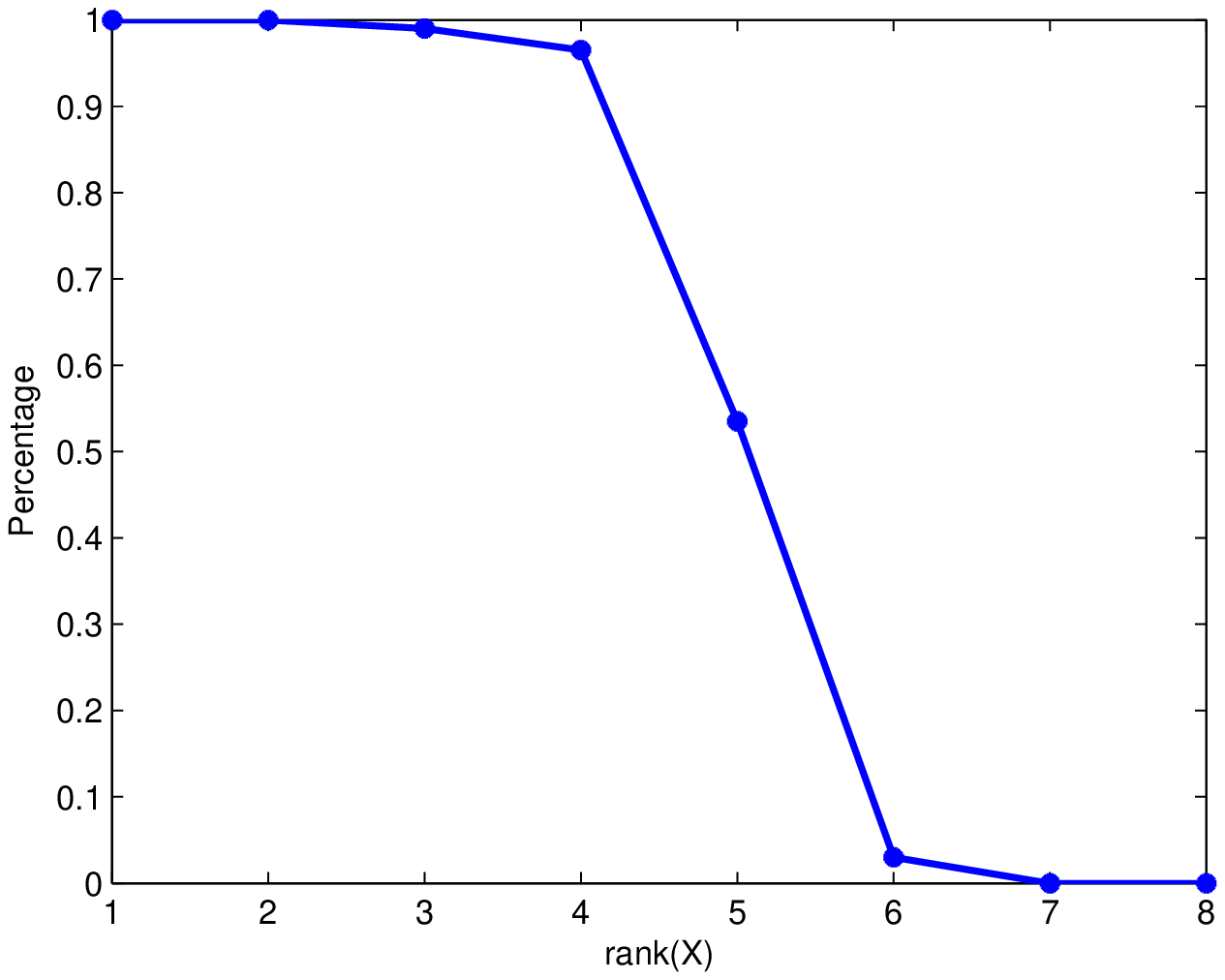}
&
\includegraphics[width=0.5\linewidth]{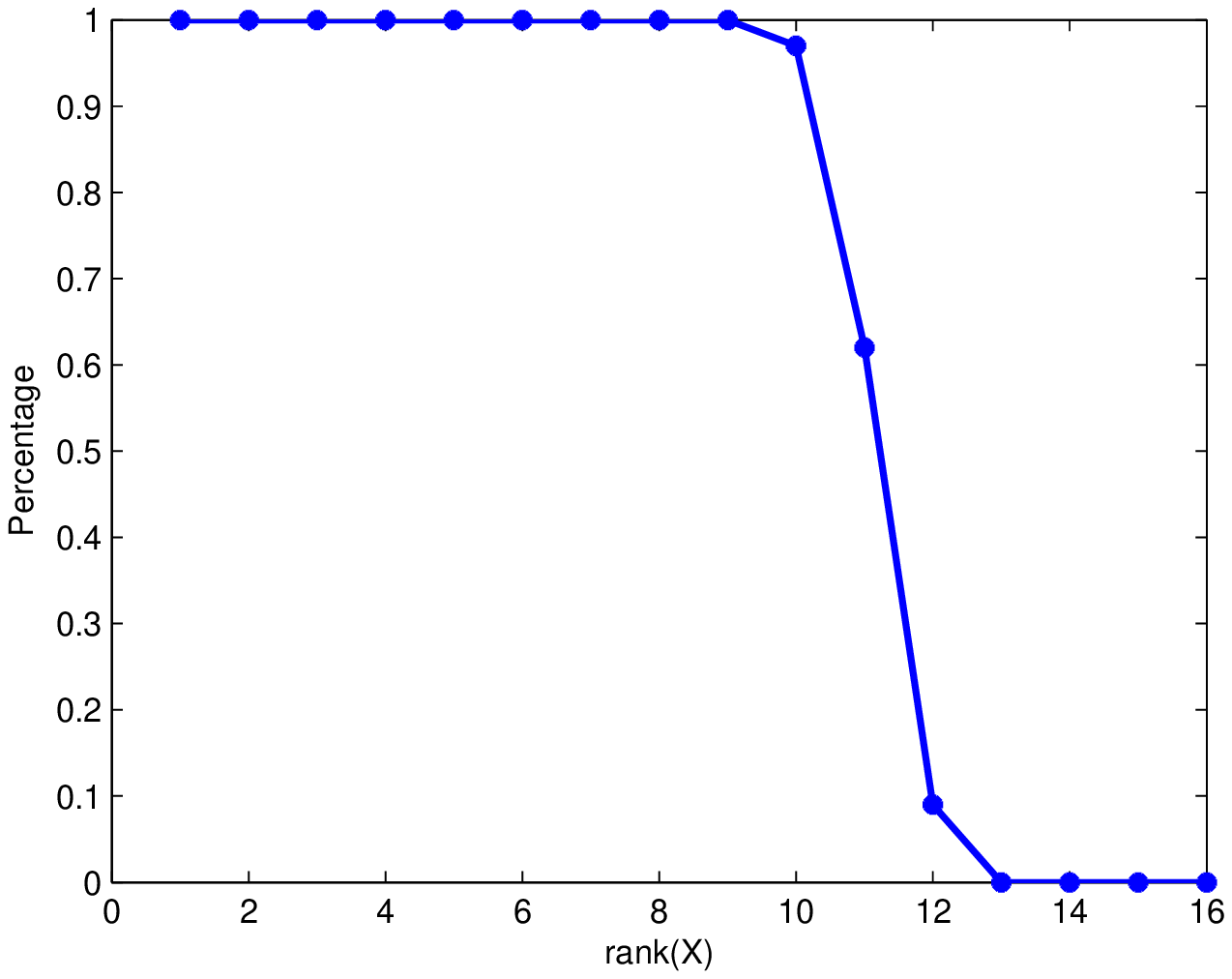}
\\
(a) {\small $m=500$ }& (b){\small $m=600$}
\end{tabular}
\caption{System of $m$ measurements admitting a unique $40 \times
40$ semidefinite matrix solution (a) $m=500$ (b) $m=600$}
      \label{fig:gaussian}

   \end{figure*}


 In the vector case, we generate a random 0-1 matrix $A^{m \times
 n}$
with i.i.d. entries and empirically study the uniqueness property
and the success of $L_1$ minimization for nonnegative vectors with
different sparsity. Each entry of $A$ takes value 1 with probability
0.2 and value 0 with probability 0.8. The size of $A$ is 50 $\times$
200 and 100 $\times$ 200 respectively. For a sparsity $k$, we select
a support set $S$ with size $|S|=k$ uniformly at random, and
generate a nonnegative vector $x_0$ on $S$ with i.i.d. entries
uniformly on the unit interval. Then we check whether $U\triangleq
\{x ~|~Ax=Ax_0, x\geq 0\}$ is singleton. This can be realized as
follows. We minimize and maximize the same objective function $d^Tx$
over $U$, where $d$ is a random vector in $\mathbb{R}^n$. Note that
if $U$ is not a singleton, then the set $\{ d \in \mathbb{R}^n~|~
d^Tx=d^Tx_0, \forall x \in U \}$ has measure 0. Thus the probability
that the minimizer and the maximizer are the same when $U$ is not a
singleton is 0. We generate several different $d$'s and claim $U$ to
be singleton if the minimizer and the maximizer
are the same for every $d$. 
For each instance, we also check whether $L_1$ minimization can
recover $x_0$ from $Ax_0$ or not. Under a given sparsity $k$, we
generate 200 $x_0$'s and repeat the above procedure 200 times.

We fix $n=200$, and $m$ is 50 in Fig. \ref{fig:result}(a) and 100 in
Fig. \ref{fig:result}(b). When $\frac{m}{n}$ increases from
$\frac{1}{4}$ to $\frac{1}{2}$, the support size of a sparse vector
which is a unique nonnegative solution increases from $0.05n$ to
$0.19n$. Note that when $\frac{m}{n}=\frac{1}{2}$, for this 0-1
matrix, the singleton property still exists linearly in $n$, while
for a random Gaussian matrix, with overwhelming probability no
vector can be a unique nonnegative solution. Besides, the thresholds
where the singleton property breaks down and where the fully
recovery of $L_1$ minimization breaks down are quite close.

In the matrix case, we generate a $40 \times 40$ matrix $G$ such
that all the elements are i.i.d. $N(0,1)$, then
$A=\frac{1}{2}(G+G^T)$ has its diagonal elements distributed as
$N(0,1)$ and off-diagonal elements distributed as
$N(0,\frac{1}{2})$. We generate $m$ such matrices $A_i$'s as the
linear operator $\cA$, $m$ is 500 and 600 respectively for
comparison. $X$ is a low-rank positive semidefinite symmetric
matrix. We increase the rank of $X$ from 0 to $0.4n$, and for each
fixed rank, generate 200 $X$'s randomly. For each $X$, we minimize
and maximize the same objective function $\langle D,X'\rangle$ over
the set $V\triangleq\{X'~|~ \cA(X')=\cA(X), X'\succeq 0, X'\in
S^n\}$, where $D$ is random matrix with i.i.d. $N(0,1)$ entries.
Similarly to the vector case, if $V$ is not a singleton, then the
set $\{ D ~|~ \langle D,X'\rangle=\langle D,X\rangle, \forall X' \in
V \}$ has measure 0. Thus the probability that the minimizer and the
maximizer are the same when $V$ is not a singleton is 0. We generate
several different $D$'s and claim the set $V$ to be a singleton if
the minimizer and the maximizer of $\langle D,X'\rangle$ from the
set $\{X'~|~ \cA(X')=\cA(X), X'\succeq 0, X'\in S^n\}$ are the same
for every $D$. As indicated by Fig. \ref{fig:gaussian}, when
$m=500$, the singleton property holds if $\rank(X)$ is at most 2,
which is $0.05n$. When $m$ increases to 600, the singleton property
holds if $\rank(X)$ is at most 8, which is $0.2n$.


\section{Conclusion}\label{sec:conclusion}
This paper studies the phenomenon that an underdetermined system
admits a unique nonnegative vector solution or a unique positive
semidefinite matrix solution. This uniqueness property can potentially lead to more
efficient sparse recovery algorithms. 
We show that only for a class of matrices with a row span intersecting the
positive orthant that $\{x ~|~Ax=Ax_0, x\geq 0\}$ could possibly be
a singleton if $x_0$ is sparse enough. Among these matrices, we are
interested in 0-1 matrices which fit the setup of network inference
problems.
For Bernoulli 0-1 matrices, we prove that with high probability the
unique solution property holds for all $k$-sparse nonnegative
vectors where $k$ is $O(n)$, instead of the previous result
$O(\sqrt{n})$. For
the adjacency matrix of a general expander, 
the same phenomenon exists 
and we further provide a closed-form constant ratio of $k$ to $n$.
One future direction is to obtain uniqueness property threshold for
a given measurement
matrix. 

For the matrix case, we develop a necessary and sufficient condition
for a linear compressed operator to admit a unique feasible positive
semidefinite matrix solution. We further show that this condition
will be satisfied with overwhelmingly high probability for a
randomly generated Gaussian linear compressed operator with vastly
different approaches from those used in vector case.  
Computing explicitly the threshold $\xi$ as a function of $\alpha$,
for the uniqueness property to happen will be one part of future
works.



\bibliographystyle{IEEEtranS}

\end{document}